\documentclass{article}

\usepackage[english]{babel}
\usepackage{amsfonts}
\usepackage{amsmath}
\usepackage{amsthm}
\usepackage{amssymb}
\usepackage{algpseudocode}

\theoremstyle{plain}
\newtheorem{theorem}{Theorem}[section]
\newtheorem{lemma}{Lemma}[section]
\newtheorem{corollary}{Corollary}[section]

\theoremstyle{definition}
\newtheorem{example}{Example}[section]

\begin{document}

\title{Finding Distinct Subpalindromes Online}

\author{
  Dmitry Kosolobov \\ dkosolobov@mail.ru \and 
  Mikhail Rubinchik \\ mikhail.rubinchik@gmail.com \and 
  Arseny M. Shur \\ arseny.shur@usu.ru}

\maketitle
\begin{center}
{Institute of Mathematics and Computer Science, Ural Federal University,\\
Ekaterinburg, Russia}
\end{center}

\begin{abstract}
We exhibit an online algorithm finding all distinct palindromes inside a given string in time $\Theta(n\log|\Sigma|)$ over an ordered alphabet and in time $\Theta(n|\Sigma|)$ over an unordered alphabet. Using a reduction from a dictionary-like data structure, we prove the optimality of this algorithm in the comparison-based computation model.


\end{abstract}

\algtext*{EndWhile}% Remove "end while" text
\algtext*{EndIf}% Remove "end if" text
\algtext*{EndFunction}% Remove "end function" text
\algtext*{EndProcedure}% Remove "end procedure" text
\algtext*{EndFor}% Remove "end for" text

\section{Introduction}

A \emph{palindrome} is a string that is equal to its reversal. Palindromes are among the most interesting text regularities. During the last few decades, many algorithmic problems concerning palindromes were considered. In this paper we solve one problem that remained open.

There is a well known online algorithm by Manacher \cite{Man} that finds all maximal subpalindromes of a string in linear time and linear space (by a ``subpalindrome'' we mean a substring that is a palindrome). It is known \cite{DJP} that every string of length $n$ contains at most $n{+}1$ distinct subpalindromes, including the empty string. The following question arises naturally: \emph{can one find all distinct subpalindromes of a string in linear time and space?} In~\cite{GPR}, this question was answered in the affirmative, but with an offline algorithm. The authors stated the existence of the corresponding online algorithm as an open problem. Our main contribution is the following result.

\begin{theorem} \label{main}
Let $\Sigma$ be a finite unordered (resp., ordered) alphabet.
There exists an online algorithm which finds all distinct subpalindromes in a string over $\Sigma$ in $O(n|\Sigma|)$ (resp., $O(n\log|\Sigma|)$) time and linear space. This algorithm is optimal in the comparison based computation model.
\end{theorem}

As a by-product, we get an online linear time and space algorithm that finds, for all prefixes of a string, the lengths of their maximal suffix-palindromes and of their palindromic closures.

\section{Notation and Definitions}

An alphabet $\Sigma$ is a finite set of letters. A \emph{string} $w$ over $\Sigma$ is a finite sequence of letters. It is convenient to consider a string as a function $w:\{1,2,\ldots,l\} \to \Sigma$. A \emph{period} of $w$ is any period of this function. The number $l$ is the \emph{length} of $w$, denoted by $|w|$. We write $w[i]$ for the $i$-th letter of $w$ and abbreviate $w[i]w[i{+}1]\cdots w[j]$ by $w[i..j]$.  A \emph{substring} of $w$ is any string $u$ such that $u=w[i..j]$ for some $i$ and $j$. Each occurrence of the substring $u$ in $w$ is determined by its \emph{position} $i$. If $i=1$ (resp. $j=|w|$), then $u$ is a \emph{prefix} (resp. \emph{suffix}) of $w$. A prefix (resp. suffix) of a string $w$ is called \emph{proper} if it is not equal to $w$. The string $w[|w|]w[|w|{-}1]\cdots w[1]$ is the \emph{reversal} of $w$, denoted by $\overleftarrow{w}$. A string is a \emph{palindrome} if it coincides with its reversal. A palindrome of even (resp. odd) length is referred to as an \emph{even} (resp. \emph{odd}) palindrome. If a substring, a prefix or a suffix of a string is a palindrome,  we call it a \emph{subpalindrome}, a \emph{prefix-palindrome}, or a \emph{suffix-palindrome}, respectively. The \emph{palindromic closure} of a string $w$ is the shortest palindrome $w'$ such that $w$ is a prefix of $w'$.

Let $w[i..j]$ be a subpalindrome of $w$. The number $\lfloor(i{+}j)/2\rfloor$ is the \emph{center} of $w[i..j]$, and the number $\lfloor(j{-}i{+}1)/2\rfloor$ is the \emph{radius} of $w[i..j]$. Thus, a single letter and the empty string are palindromes of radius 0. Note that the center of the empty subpalindrome is the previous position of the string.

By an \emph{online algorithm} for an algorithmic problem concerning strings we mean an algorithm that processes the input string $w$ sequentially from left to right, and answers the problem for each prefix $w[1..j]$ of $w$ after processing the letter $w[j]$.

\section{Distinct subpalindromes}
\subsection{Suffix-Palindromes and Palindromic Closure} \label{ssec:closure}

The problem of finding the lengths of palindromic closures for all prefixes of a string is closely related to the problem of finding all distinct subpalindromes of this string. It was conjectured in \cite{GPR} that there exists an online linear time algorithm for the former problem.

Let $v$ be the maximal suffix-palindrome of $w = uv$. It is easy to see that the palindromic closure of $w$ equals to the string $uv\overleftarrow{u}$. An offline algorithm for finding the maximal suffix-palindromes for each prefix of the string can be found, e. g., in \cite[Ch. 8]{CrRy}. Our online algorithm is a modification of Manacher's algorithm (see \cite{Man}).

We construct a data structure based on Manacher's algorithm. Let $\Delta$ be a boolean flag (needed to distinguish between odd and even palindromes). This data structure $\mathrm{man}$ contains a string $text$ and supports the procedure $\mathrm{man.AddLetter}(c)$ adding a letter to the end of $text$. The function $\mathrm{man.MaxPal}$ returns the length of maximal odd/even (according to $\Delta = 0/1$) suffix-palindrome of $text$.

Our data structure uses the following internal variables:\\
$n$, which is the length of $text$;\\
$i$, which is the center of the maximal odd/even (according to $\Delta = 0/1$) suffix-palindrome of $text$;\\
$Rad$, which is an array of integers such that for any $j<i$ the value $Rad[j]$ is equal to the radius of the maximal odd/even (according to $\Delta = 0/1$) subpalindrome with the center $j$. The main property of $Rad$ is expressed in the following lemma (see \cite[Lemma 8.1]{CrRy}).

\begin{lemma} \label{ManacherLemma}
Let $k$  be an integer, $1\le k\le Rad[i]$.\\
(1) If $Rad[i{-}k] < Rad[i] - k$ then $Rad[i{+}k] = Rad[i{-}k]$;\\
(2) if $Rad[i{-}k] > Rad[i] - k$ then $Rad[i{+}k] = Rad[i] - k$.
\end{lemma}

At the beginning, $Rad$ is filled with zeros, $n = 1$, $i = 2$, $text = "\$"$, where $\$$ is a special letter that does not appear in the input string\begin{footnote}{The strange-looking initial value of $i$ provides the correct processing of the first letter after $\$$ (the while loop will be skipped and the correct values $n=i=2$ for the next iteration will be obtained).}\end{footnote}.

\begin{algorithmic}[1]
	\Procedure{man.AddLetter}{$c$}
		\State $s \gets i - Rad[i] + \Delta$\Comment{position of the max suf-pal of $text[1..n]$} \label{lst:line:calcs}
        \State $text[n+1] \gets c$
		\While{$i + Rad[i] \leqslant n$}\label{lst:line:forbeg}
			\State $Rad[i] \gets \min(Rad[s{+}n{-}i{-}\Delta], n - i)$\Comment{this is $Rad[i]$ in $text[1..n]$}\label{lst:line:min}
			\If{$i + Rad[i] = n \mathrel{\mathbf{and}} text[i{-}Rad[i]{-}1{+}\Delta] = c$}\label{lst:line:breadth}
				\State $Rad[i] \gets Rad[i] + 1$\Comment{extending the max suf-pal}
				\State $\mathbf{break}$     \Comment{max suf-pal of $text[1..n{+}1]$ found}
			\EndIf
            \State $i \gets i + 1$ \Comment{next candidate for the center of max suf-pal}
		\EndWhile\label{lst:line:forend}
		\State $n \gets n + 1$
	\EndProcedure
	\Function{man.MaxPal}{}
		\State \Return $2 Rad[i] + 1 - \Delta$\label{lst:line:maxpal}
	\EndFunction
\end{algorithmic}

\begin{theorem} \label{PalindromeClosure}
There exists an online linear time and space algorithm that finds the lengths of the maximal suffix-palindromes of all prefixes of a string.
\end{theorem}

\begin{proof}
From the correctness of Manacher's algorithm (see \cite{Man}) and Lemma~\ref{ManacherLemma} it follows that the function $\mathrm{man.MaxPal}$ correctly returns the length of the maximal odd/even suffix palindrome of the processed string. For a string of length $n$, we call the procedure $\mathrm{man.AddLetter}$\ $n$ times with the parameter $\Delta=0$ and $n$ times with $\Delta=1$. If one call of the procedure uses $k$ iterations of the loop in the lines \ref{lst:line:forbeg}--\ref{lst:line:forend}, then the value of $i$ increases by $k{-}1$. Hence, the loop is used at most $4n$ times in total. Apart from this loop, $\mathrm{man.AddLetter}$ performs a constant number of operations. This gives us the required $O(n)$ time bound.
\end{proof}

\begin{corollary}
There exists an online linear time and space algorithm that finds the lengths of palindromic closured of all prefixes of a string.
\end{corollary}

\begin{example}
Let $w = abadaadcaa$ and consider the state of the data structure $\mathrm{man}$ after the sequence of calls $\mathrm{man.AddLetter}(w[i])$, $i = 1,2,\ldots,10$.
$$
\begin{array}{l}
text=\$w;\\
Rad=(0, 1, 0, 1, 0, 0, 0, 0, 0, 0)\text{ for }\Delta=0;\\
Rad=(0, 0, 0, 0, 2, 0, 0, 0, 1, 0)\text{ for }\Delta=1;\\
\end{array}
$$
The calls to $\mathrm{man.MaxPal}$ after each call to $\mathrm{man.AddLetter}(w[i])$ return consequently the values $1, 1, 3, 1, 3, 1, 1, 1, 1, 1$ for the case $\Delta = 0$ and $0, 0, 0, 0, 0, 2, 4, 0, 0, 2$ for the case $\Delta = 1$.
\end{example}

\subsection{Distinct subpalindromes} \label{ssec:distinct}

We make use of the following

\begin{lemma}[\cite{GPR}]
Each subpalindrome of a string is the maximal suffix-palindrome of some prefix of this string.
\end{lemma}

This lemma implies that the online algorithm designed in Sect.~\ref{ssec:closure} finds all subpalindromes of a string. To find all distinct subpalindromes, we have to verify whether the maximal suffix-palindrome of a string has another occurrence in this string. Note that the direct comparison of substrings for this purpose leads to at least quadratic overall time. Instead, we will use a version of suffix tree known as \emph{Ukkonen's tree}. To introduce it, we need some definitions.

A \emph{trie} is a rooted labelled tree in which every edge is labelled with a letter such that all edges leading from a vertex to its children have different labels. Each vertex of the trie is associated with the string labelling the path from the root to this vertex. A trie can be ``compressed'' as follows: any non-branching descending path is replaced by a single edge labelled by the string equal to the label of this path. The result of this procedure is called a \emph{compressed trie}. For a set $S$ of strings, the \emph{compressed trie of} $S$ is defined by the following two properties: (i) for each string of $S$, there is a vertex associated it and (ii) the trie has the minimal number of vertices among all compressed tries with property (i).

A (compressed) \emph{suffix tree} is the compressed trie of the set of all suffixes of a string. \emph{Ukkonen's tree} is the data structure $\mathrm{ukk}$ containing a string and the suffix tree of this string (labels are stored as pairs of positions in the string). Ukkonen's tree allows one to add a letter to the end of the string (procedure $\mathrm{ukk.addLetter}(c)$), updating the suffix tree. We also need the following parameter: the length of the minimal suffix of the processed string such that this suffix occurs in this string only once (function $\mathrm{ukk.minUniqueSuff}$). Let us recall some implementation details of Ukkonen's tree for the efficient implementation of $\mathrm{ukk.minUniqueSuff}$.

The update of Ukkonen's tree is based on the system of suffix links. Such a link connects a vertex associated with a word $v$ to the vertex associated with the longest proper suffix of $v$. These links are also defined for ``implicit'' vertices (the vertices that are not in the compressed trie, but present in the corresponding trie). In particular, Ukkonen's tree supports the triple $(v,e,i)$ such that

\tabcolsep=3pt
\begin{tabular}{rl}
(1)&$v$ is a vertex (associated with some string $s'$) of the current suffix tree,\\
(2)&$e$ is an edge (labelled by some string $s$) between $v$ and its child,\\
(3)&$i$ is an integer between $0$ and $|s|$,
\end{tabular}

\noindent with the property that $s's[1..i]$ is the longest suffix of the processed string that occurs in this string at least twice. This triple is crucial for  fast update of Ukkonen's tree (for further details, see \cite{Ukk}).

\begin{lemma}[\cite{Ukk}] \label{lemukk}
The procedure $\mathrm{ukk.addLetter}(c)$ performs $n$ calls using $O(n)$ space and $O(n\log|\Sigma|)$ (resp., $O(n|\Sigma|)$) time in the case of ordered (resp., unordered) alphabet.
\end{lemma}

We modify Ukkonen's tree, associating with each vertex $u$ an additional field $u.\mathrm{depth}$ to store the length of the string associated with $u$. Maintaining this field requires a constant number of operations at the moment when $u$ is created. Thus, this update adds $O(n)$ time and $O(n)$ space to the total cost of maintaining Ukkonen's tree. Thus, Lemma~\ref{lemukk} holds for the modified Ukkonen's tree as well. It remains to note that $\mathrm{ukk.minUniqueSuff}=v.\mathrm{depth}+i+1$.

\begin{proof}[Theorem~\ref{main}: existence]
The following algorithm solves the problem and has the required complexity. The algorithm uses data structures $\mathrm{man}$ and $\mathrm{ukk}$, processing the same input string $w$. The next (say, $n$th) symbol of $w$ is added to both structures through the procedures $\mathrm{man.AddLetter}$ and $\mathrm{ukk.AddLetter}$. After this, we call $\mathrm{man.MaxPal}$ to get the length of the maximal palindromic suffix of $w[1..n]$ and $\mathrm{ukk.MinUniqueSuff}$ to get the length of the shortest suffix of $w[1..n]$ that never occurred in $w$ before. The inequality $\mathrm{man.MaxPal}\ge \mathrm{ukk.MinUniqueSuff}$ means the detection of a new palindrome; we get its first and last positions from the structure $\mathrm{man}$ and output them. In the case of the opposite inequality, there is no new palindrome, and we output ``---''.

The required time and space bounds follow from Theorem~\ref{PalindromeClosure} and Lemma~\ref{lemukk}.
\end{proof}

\begin{example}
Consider the string $w = abadaadcaa$ again. We get the following results for $i = 1,2,\ldots,10$:
$$
\arraycolsep=3pt
\begin{array}{lllllllllll}
\mathrm{man.MaxPal}:&1&1&3&1&3&2&4&1&1&2\\
\mathrm{ukk.MinUniqueSuff}:&1&1&2&1&2&2&3&1&2&3\\
\mathrm{output}:&1{-}1&2{-}2&1{-}3&4{-}4&3{-}5&5{-}6&4{-}7&8{-}8&\text{---}&\text{---}
\end{array}
$$
\end{example}

\subsection{Lower bounds}

Recall that a \emph{dictionary} is a data structure $D$ containing some set of elements and designed for the fast implementation of basic operations like checking the membership of an element in the set, deleting an existing element, or adding a new element. Below we consider an \emph{insert-only dictionary over a set $S$}. In each moment, such a dictionary $D$ contains a subset of $S$ and supports only the operation $\mathrm{insqry}(x)$. This operation checks whether the element $x\in S$ is already in the dictionary; if no, it adds $x$ to the dictionary.

\begin{lemma} \label{DictionaryLowerBound}
Suppose that the alphabet $\Sigma$ consists of indivisible elements, $n\ge |\Sigma|$, and the insert-only dictionary $D$ over $\Sigma$ is initially empty. Then the sequence of $n$ calls of $\mathrm{insqry}$ requires, in the worst case, $\Omega(n\log|\Sigma|)$ time if $\Sigma$ is ordered and $\Omega(n|\Sigma|)$ if $\Sigma$ is unordered.
\end{lemma}

\begin{proof}
Let $\Sigma = \{a_1<a_2<\ldots<a_m\}$ be an ordered alphabet. Assume that on some stage all letters with even numbers are in the dictionary, while all elements with odd numbers are not. Consider the next operation. In the comparison-based computation model, a query ``$x\in D$?'' is answered by some decision tree; each node of this tree is marked by the condition ``$x<a_i$'' for some $i$. To distinguish between $a_i$ and $a_{i+1}$, the tree should contain the nodes for both $a_i$ and $a_{i+1}$. Now note that for any $i$, exactly one of the letters $a_i$ and $a_{i+1}$ belongs to $D$. So, to answer correctly all possible queries ``$x\in D$?'' the decision tree should have nodes for all letters. Then the depth of this tree is $\Omega(\log m)$. Therefore, for some element $x=a_{2i}$ the number of comparisons needed to prove that $x\in D$ is $\Omega(\log m)$. After processing $x$, the content of the dictionary remains unchanged. The decision tree can change, but it does not matter: we again choose the next letter to be the one having an even number and requiring $\Omega(\log m)$ comparisons to prove its membership in $D$. Thus, our ``bad'' sequence of calls is as follows: it starts with $\mathrm{insqry}(a_2),\ldots,\mathrm{insqry}(a_{2\lfloor m/2\rfloor})$, and continues with the ``worst'' letter, described above, on each next step. Even if the first $\lfloor m/2\rfloor$ calls can be performed in $O(1)$ time each, the overall time is $\Omega(n\log m)$, as required.

In the case of unordered alphabet all conditions in the decision tree have the form ``$x = a_i$''. It is clear that if the dictionary contains $\lfloor m/2\rfloor$ elements, the maximal number of comparisons equals $\lfloor m/2\rfloor$ as well. Choosing the bad sequence of calls in the same way as for the ordered alphabet, we arrive at the required bound $\Omega(nm)$.
\end{proof}

Before finishing the proof of Theorem~\ref{main} we mention the following lemma. Its proof is obvious.

\begin{lemma} \label{abx}
Suppose that $a,b$ are two different letters and $w=abx_1abx_2\cdots abx_n$ is a string such that each $x_i$ is a letter different from $a$ and $b$. Then all nonempty subpalindromes of $w$ are single letters.
\end{lemma}

\begin{proof}[Proof of Theorem~\ref{main}: lower bounds]
We prove the required lower bounds reducing the problem of maintaining an insert-only dictionary to counting distinct palindromes in a string. Assume that we have a black box algorithm that processes an input string letter by letter and outputs, after each step, the number of distinct palindromes in the string read so far. The time complexity of this algorithm depends on the length $n$ of the string at least linearly, and a linear in $n$ algorithm does exist, as we have proved in the Sect.~\ref{ssec:distinct}. Thus, we can assume that the considered black box algorithm works in time $O(n\cdot f(m))$, where $m$ is the size of the alphabet of the processed string and the function $f(m)$ is non-decreasing.

The insert-only dictionary over a set $\Sigma$ of size $m>1$ can be maintained as follows. We pick up two letters $a,b\in\Sigma$ and mark their presence in the dictionary using two boolean variables, $z_a$ and $z_b$. All other letters are processed with the aid of the mentioned black box. Let us describe how to process a sequence of $n$ calls $\mathrm{insqry}(x_1),\ldots,\mathrm{insqry}(x_n)$ starting from the empty dictionary.

For each call, we first compare the current letter $x_i$ to $a$ and $b$. If $x_i=a$, then $z_a$ is the answer to the query ``$x_i\in D$?''; after answering the query we set $z_a=1$. The case $x_i=b$ is managed in the same way.

If $x_i\notin\{a,b\}$, we feed the black box with $a$, $b$, and $x_i$ (in this order). Then we get the output of the black box and check whether the number of distinct subpalindromes in its input string increased. By Lemma~\ref{abx}, the increase happens if and only if $x_i$ appears in the input string of the black box for the first time. Thus, we can immediately answer the query ``$x_i\in D$?'', and, moreover, $x_i$ is now in the dictionary.

The described algorithm performs the sequence of calls $\mathrm{insqry}(x_1),\ldots,\mathrm{insqry}(x_n)$ in time $O(n)$ plus the time used by the blackbox to process a string of length $\le 3n$ over $\Sigma$. Hence, the overall time bound is $O(n\cdot f(m))$. In view of Lemma~\ref{DictionaryLowerBound} we obtain $f(m)=\Omega(\log m)$ (resp., $f(m)=\Omega(m)$) in the case of ordered (resp., unordered) alphabet $\Sigma$. The required lower bounds are proved.
\end{proof}

\bibliographystyle{amsplain}          % please use the PSC bib-style
\bibliography{unique_palindromes_eng2}            % referring sample.bib

\end{document}